\documentclass[10pt]{article}
\usepackage{csquotes}
\usepackage{tikz}

\usepackage{amsthm,amsmath,amssymb}
\usepackage{graphicx}

\usepackage[colorlinks=true,citecolor=black,linkcolor=black,urlcolor=blue]{hyperref}

\usepackage[linesnumbered,ruled,vlined]{algorithm2e}
\usepackage{dsfont}
\usepackage[top=2.5cm, bottom=3cm, left=2.75cm, right=2.75cm]{geometry}
\theoremstyle{plain}
\newtheorem{theorem}{Theorem}
\newtheorem{lemma}[theorem]{Lemma}

\theoremstyle{definition}

\theoremstyle{remark}

\newcommand{\OO}{{\mathcal{O}}}
\newcommand{\m}[1]{\mathrm{#1}}
\newcommand{\black}{\m{black}}
\newcommand{\classNP}{\mathcal{NP}}
\newcommand{\classP}{\mathcal{P}}

\title{\bf On the Query Complexity of Black-Peg AB-Mastermind}

\author{Mourad El Ouali \, Christian Glazik \, Volkmar Sauerland and Anand Srivastav\\
\small Department of Computer Science\\[-0.8ex]
\small Christian-Albrechts-Universit\"{a}t zu Kiel, Kiel, Germany\\[-0.8ex]
\small\tt <cgl,meo,vsa,asr>@informatik.uni-kiel.de\\
}

\date{}

\begin{document}

\maketitle


\begin{abstract}
\emph{Mastermind} game is a two players zero sum game of imperfect information.
The first player, called \enquote{codemaker},
chooses a secret code and the second player, called \enquote{codebreaker}, 
tries to break the secret code by making as few guesses as possible,
exploiting information that is given by the codemaker after each guess.
In this paper, we consider the so called Black-Peg variant of Mastermind,
where the only information concerning a guess 
is the number of positions in which the guess coincides with the secret code.
More precisely, we deal with a special version of the Black-Peg game with $n$ holes and $k\ge n$ colors where no repetition of colors is allowed.
We present upper and lower bounds on the number of guesses necessary to break the secret code.
We first come back to the upper bound results introduced by El Ouali and Sauerland (2013).
For the case $k=n$ the secret code can be algorithmically identified within less
than $(n-3)\lceil\log_2{n}\rceil +\frac{5}{2}n$ queries.
That result improves the result of Ker-I Ko and Shia-Chung Teng (1985) by almost a factor of 2.
For the case $k>n$ we prove an upper bound for the problem of $(n-2)\lceil\log_2{n}\rceil+k+1$.
Furthermore we prove a new lower bound 
for (a generalization of) the case $k=n$ that improves the recent result of Berger et al. (2016) from $n-\log\log(n)$ to $n$.
We also give a lower bound of $k$ queries for the case $k>n$.

\end{abstract}

\section{Introduction}
In this paper we deal with \emph{Mastermind}, which is a popular board game that in the past three decades has become interesting from the algorithmic point of view.
Mastermind is a two players board game invented in 1970 by the postmaster and telecommunication expert Mordecai Meirowitz.
The idea of the game is that the codemaker chooses a secret color combination of $n$ pegs from $k$ possible colors and the codebreaker has to identify 
the code by a sequence of queries and corresponding information that is provided by the codemaker.
All queries are also color combinations of $n$ pegs.
Information is given about the number of correctly positioned colors and further correct colors, respectively.
Mathematically, the codemaker selects a vector $y\in [k]^n$ and the codebreaker gives in each iteration a query in form of a vector $x\in [k]^n$.
The codemaker replies with a pair of two numbers, called $\m{black}(x,\,y)$ and $\m{white}(x,\,y)$, respectively.
The first one is the number of positions in which both vectors $x$ and $y$ coincide and the second one is the number of additional pegs with a right color but a wrong position:
\begin{align*}
\m{black}(x,y)&=|\{i\in [n];\, x(i)=y(i) \}|,\\
\m{white}(x, y)&=\max_{\sigma\in S_{n}}|\{i\in [n];\, y(i)=x(\sigma(i)) \}|\\
&\phantom{=}\; -\m{black}(x, y).
\end{align*}
The Black-Peg game is a special version of Mastermind, where answers are provided by $\m{black}$ information, only.
A further version is the so-called AB game in which all colors within a code must be distinct.
In this paper, we deal with a special combination of the Black-Peg game and the AB game, where both the secret vector and 
the guesses must be composed of pairwise distinct colors ($k\ge n$) and the answers are given by the $\m{black}$ information, only.\par

{\bf Related Works:} 
In 1963, several years before the invention of Mastermind as a commercial board game, Erd\"os and R\'enyi \cite{ER63} analyzed the same problem with two colors.
One of the earliest analysis of this game after its commercialization dealing with the case of 4 pegs and 6 colors was done by Knuth \cite{K77}.
He presented a strategy that identifies the secret code in at most 5 guesses.
Ever since the work of Knuth the general case of arbitrary many pegs and colors has been intensively investigated in combinatorics and computer science literature.
In the field of complexity, Stuckman and Zhang \cite{SZ06} showed that it is ${\classNP}$-complete to determine if a sequence of queries and answers is satisfiable. 
Concerning the approximation aspect, there are many works regarding different methods \cite{BGL09,CCH96,C83,DW12,FL10,GCG11,GMC11,Goo09a,JP09,KC03,KL93,SZ06}.
The Black-Peg game was first introduced by Chv\'atal for the case $k=n$.
He gave a deterministic adaptive strategy that uses $2n\lceil\log_2{k}\rceil + 4n$ guesses.
Later, Goodrich \cite{Goo09b} improved the result of Chv\'atal for arbitrary $n$ and $k$ to $n\lceil\log_2{k}\rceil + \lceil (2-1/k)n \rceil + k$ guesses.
Moreover, he proved in the same paper that this kind of game is $\classNP$-complete.
A further improvement to $n\lceil\log_2{n}\rceil + k - n +1$ for $k>n$ and $n\lceil\log_2{n}\rceil + k$ for $k\le n$ was done by J\"ager and Peczarski \cite{JP11}.
Recently, Doerr et al. \cite{DSTW13} improved the result obtained by Chv\'atal to $\OO(n \log \log n)$ and also showed that this asymptotic order even holds for up to $n^2\log\log n$ colors, 
if both black and white information is allowed.
For the AB game J\"ager and Peczarski \cite{JP15} proofed exact worst-case numbers of guesses for fixed $n\in\{2,3,4\}$
and arbitrary $k$.
Concerning the combination of both variants, Black-Peg game and AB game, for almost 3 decades the work due to Ker-I Ko and Shia-Chung Teng \cite{KT86} was the only contribution that provides an upper bound for the case $k=n$.
They presented a strategy that identifies the secret permutation in at most $2n\log_2{n} + 7n$ guesses and proved that the corresponding counting problem is $\# {\classP}$-complete.

{\bf Our Contribution:} In this paper we consider the Black-Peg game without color repetition.
We first present a polynomial-time algorithm that identifies the secret permutation in less than $n\log_2{n}+\frac{3}{2}n$ queries in the case $k=n$ and in less than $n\log_2{n}+k+2n$ queries in the case $k>n$.
The constructive strategy origined in the work of El Ouali and Sauerland \cite{ES13}.
Our result for the case $k=n$ improves the result of Ker-I Ko and Shia-Chung Teng \cite{KT86} by almost a factor of 2. 
Furthermore we analyze the worst-case performance of query strategies for both variants of the Game and give a new lower bound of $n$ queries for the case $k=n$, 
which improves the recently presented lower bound of $n-\log\log(n)$ by Berger et al \cite{BCS16}.
We note, however, that the corresponding asymptotic bound of $O(n)$ is long-established.
For $k>n$ we give a lower bound of $k$.
Both lower bounds even hold if the codebreaker is allowed to use repeated colors in his guesses.

\section{Upper Bounds on the Number of Queries}\label{section:KequalsN}

We first consider Black-Peg Mastermind with $k=n$ and the demand for pairwise distinct colors in both the secret code and all queries,
i.e., we deal with permutations in $S_n$.

\subsection{The Case $k=n$: Permutation-Mastermind}

For convenience, we will use the term permutation for both, a mapping in $S_n$ and its one-line representation as a vector.
Our algorithm for finding the secret permutation $y\in S_n$ includes two main phases which are based on two ideas.
In the first phase the codebreaker guesses an initial sequence of $n$ permutations that has a predefined structure.
In the second phase, the structure of the initial sequence and the corresponding information by the codemaker enable us to identify 
correct {\em components} $y_i$ of the secret code one after another, each by using a binary search.
Recall, that for two codes $w=(w_1,\dots,w_n)$ and $x=(x_1,\dots,x_n)$, we denote by $\m{black}(w,x)$ the number $|\{i\in[n]\,|\,w_i=x_i\}|$ of components in which $w$ and $x$ are equal.
We denote the mapping $x$ restricted to the set $\{s,\dots,l\}$ with $(x_i)_{i=s}^{l}$, $s,l\in [n]$.\par

{\bf Phase 1.} Consider the $n$ permutations, $\sigma^1,\dots,\sigma^n$, that are defined as follows: $\sigma^1$ corresponds to 
the identity map and for $j\in[n-1]$, we obtain $\sigma^{j+1}$ from $\sigma^{j}$ by a circular shift to the right.
For example, if $n=4$, we have $\sigma^1=(1,2,3,4)$, $\sigma^2=(4,1,2,3)$, $\sigma^3=(3,4,1,2)$ and $\sigma^4=(2,3,4,1)$.
Within those $n$ permutations, every color appears exactly once at every position and, thus, we have
\begin{equation}\label{infosum}
\sum_{j=1}^{n}\m{black}(\sigma^j,y)=n.
\end{equation}
The codebreaker guesses $\sigma^1,\dots,\sigma^{n-1}$ and obtains the additional information $\m{black}(\sigma^n,y)$ from (\ref{infosum}).\par

{\bf Phase 2.} The strategy of the second phase identifies the values of $y$ one after another.
This is done by using two binary search routines, called {\sc findFirst} and {\sc findNext}, respectively.
The idea behind both binary search routines is to exploit the information that for $1\leq i,j\leq n-1$ we 
have $\sigma^{j}_{i}=\sigma^{j+1}_{i+1}$, $\sigma^{n}_{i}=\sigma^{1}_{i+1}$, $\sigma^{j}_{n}=\sigma^{j+1}_{1}$ and $\sigma^{n}_{n}=\sigma^{1}_{1}$.
While, except for an unfrequent special case, {\sc findFirst} is used to identify the first correct component of the secret code, {\sc findNext} identifies the remaining components in the main loop of the algorithm.
Actually, {\sc findFirst} would also be able to find the remaining components but requires more guesses than {\sc findNext} (twice as many in the worst case).
On the other hand, {\sc findNext} only works if at least one value of $y$ is already known such that we have to identify the value of one secret code component in advance.\par

{\bf Identifying the First Component:} Equation (\ref{infosum}) implies that either $\m{black}(\sigma^j,y)=1$ holds for all $j\in[n]$ or that we can find a $j\in[n]$ with $\m{black}(\sigma^j,y)=0$.

In the first case, which is unfrequent, we can find one correct value of $y$ by guessing at most $\frac{n}{2}+1$ modified versions of some initial guess, say $\sigma^1$.
Namely, if we define a guess $\sigma$ by swapping a pair of components of $\sigma^1$, we will obtain $\m{black}(\sigma,y)=0$, if and only if one of the swapped components has the correct value in $\sigma^1$.

In the frequent second case, we find the first component by {\sc findFirst} in at most $2\lceil\log_2{n}\rceil$ guesses.
The routine {\sc findFirst} is outlined as Algorithm \ref{findFirst} and works as follows:
In the given case, we can either find a $j\in[n-1]$ with $\m{black}(\sigma^j,y)>0$ but $\m{black}(\sigma^{j+1},y)=0$ 
and set $r:=j+1$, or we have $\m{black}(\sigma^n,y)>0$ but $\m{black}(\sigma^1,y)=0$ and set $j:=n$ and $r:=1$.
We call such an index $j$ an {\em active} index.
Now, for every $l\in\{2,3,\dots,n\}$ we define the code
\[
\sigma^{j,l}:=\left((\sigma^{j}_i)_{i=1}^{l-1},\sigma^{r}_{1},(\sigma^{r}_{i})_{i=l+1}^n\right),
\]
and call the peg at position $l$ in $\sigma^{j,l}$ the pivot peg.
From the information $\sigma^{j}_{i}=\sigma^{r}_{i+1}$ for $1\leq i\leq n-1$ we conclude that $\sigma^{j,l}$ is actually a new permutation as required.
The fact that $\m{black}(\sigma^{r},y)=0$ implies that the number of correct pegs up to position $l-1$ in $\sigma^{j}$ is 
either $\m{black}(\sigma^{j,l},y)$ (if $y_l\ne \sigma^{r}_{1}$) or $\m{black}(\sigma^{j,l},y)-1$ (if $y_l= \sigma^{r}_{1}$).
For our algorithm, we will only need to know if there exist one correct peg in $\sigma^{j}$ up to position $l-1$.
The question is cleared up, if $\m{black}(\sigma^{j,l},y) \ne 1$.
On the other hand, if $\m{black}(\sigma^{j,l},y) = 1$, we can define a new guess $\rho^{j,l}$ by swapping the pivot peg with a wrong peg in $\sigma^{j,l}$.
We define
\[
\rho^{j,l}:=\begin{cases}
\left((\sigma^{j}_{i})_{i=1}^{l},\sigma^{r}_{1},(\sigma^{r}_{i})_{i=l+2}^{n}\right)& \text{if }l<n\\
\left(\sigma^{r}_{1},(\sigma^{j}_{i})_{i=2}^{n-1},\sigma^{j}_{1}\right)& \text{if }l=n
\end{cases}
\]
assuming for the case $l=n$, that we know that $\sigma^{j}_1\ne y_1$.
We will obtain $\m{black}(\rho^{j,l},y)>0$, if and only if the pivot peg had a wrong color before, meaning that there is one correct peg in $\sigma^{j}$ in the first $l-1$ places.
Thus, we can find the position $m$ of the left most correct peg in $\sigma^{j}$ by a binary search as outlined in Algorithm \ref{findFirst}.\par
\begin{algorithm}[h]
\SetKwInOut{Input}{input}\SetKwInOut{Output}{output}
\Input{Code $y$ and an active index $j\in[n]$}
\Output{Left most correct peg position in $\sigma^j$}
\lIf{$j=n$}{$r:=1$ }\;\lElse{$r:=j+1$}\;
$a:=1$\;
$b:=n$\;
$m:=n$ \tcp*{position to be found}
\While{$b>a$}
{
    $l:=\lceil \frac{a+b}{2} \rceil$ \tcp*{pivot position}
	Guess $\sigma^{j,l}:=\left((\sigma^{j}_i)_{i=1}^{l-1},\sigma^r_1,(\sigma^r_i)_{i=l+1}^{n}\right)$\;
	$s:=\m{black}(\sigma^{j,l},y)$\;
	\If{$s=1$}
	{
		 \lIf{$l<n$}{$\rho^{j,l}:=\left((\sigma^{j}_{i})_{i=1}^{l},\sigma^{r}_{1},(\sigma^{r}_{i})_{i=l+2}^{n}\right)$}\;
		\lElse{$\rho^{j,l}:=\left(\sigma^{r}_{1},(\sigma^{j}_{i})_{i=2}^{n-1},\sigma^{j}_{1}\right)$}\;
		Guess $\rho^{j,l}$\;
		$s:=\m{black}(\rho^{j,l},y)$\;
	}
	\If{$s>0$}
	{
		$b:=l-1$\;
		\lIf{$b<m$}{$m:=b$}\;
	}
	\lElse{$a:=l$}\;
}
Return $m$\;
\caption{Function {\sc findFirst}}\label{findFirst}
\end{algorithm}
{\bf Identifying a Further Component:} For the implementation of {\sc findNext} we deal with a partial solution vector $x$ that satisfies $x_i\in\{0,y_i\}$ for all $i\in[n]$.
We call the (indices of the) non-zero components of the partial solution {\em fixed}.
They indicate the components of the secret code that have already been identified.
The (indices of the) zero components are called {\em open}.
Whenever {\sc findNext} makes a guess $\sigma$, it requires to know the number of open components in which the guess coincides with the secret code, i.e. the number
\[
\m{black}(\sigma,y,x) := \m{black}(\sigma,y) - \m{black}(\sigma,x).
\]
Note, that the term $\m{black}(\sigma,x)$ is known by the codebreaker.
After the first component of $y$ has been found and fixed in $x$, there exists a $j\in [n]$ such that $\m{black}(\sigma^j,y,x)=0$.
As long as we have open components in $x$, we can either find a $j\in[n-1]$ with $\m{black}(\sigma^j,y,x)>0$ but $\m{black}(\sigma^{j+1},y,x)=0$ 
and set $r:=j+1$, or we have $\m{black}(\sigma^n,y,x)>0$ but $\m{black}(\sigma^1,y,x)=0$ and set $j:=n$ and $r:=1$.
Again, we call such an index $j$ an {\em active} index.
Let $j$ be an active index and $r$ its related index.
Let $c$ be the color of some component of $y$ that is already identified and fixed in the partial solution $x$.
With $l_{j}$ and $l_{r}$ we denote the position of color $c$ in $\sigma^{j}$ and $\sigma^{r}$ respectively.
The peg with color $c$ serves as a pivot peg for identifying a correct position $m$ in $\sigma^j$ that is not fixed, yet.
There are two possible modes for the binary search that depend on the fact if $m\le l_j$.
The mode is indicated by a Boolean variable $\m{leftS}$ and determined by lines 4 to 8 of {\sc findNext}.
Clearly, $m\le l_j$ if $l_j=n$.
Otherwise, the codebreaker guesses
\[
\sigma^{j,0}:=\left(c,(\sigma^{j}_i)_{i=1}^{l_{j}-1},(\sigma^{j}_i)_{i=l_{j}+1}^{n}\right),
\]
By the information $\sigma^{j}_{i}=\sigma^{r}_{i+1}$ we obtain that $(\sigma^{j}_i)_{i=1}^{l_{j}-1}\equiv (\sigma^{r}_{i})_{i=2}^{l_{j}}$.
We further know that every open color has a wrong position in $\sigma^r$.
For that reason, $\m{black}(\sigma^{j,0},y,x)=0$ implies that $m\le l_j$.

\begin{algorithm}[h]
\SetKwInOut{Input}{input}\SetKwInOut{Output}{output}
\Input{Code $y$, partial solution $x\ne 0$ and an active index $j\in[n]$}
\Output{Position $m$ of a correct open component in $\sigma^j$}
\lIf{$j=n$}{$r:=1$ }\;\lElse{$r:=j+1$}\;
Choose a color $c$ with identified position (a value $c$ of some non-zero component of $x$)\;
Let $l_j$ and $l_r$ be the positions with color $c$ in $\sigma^j$ and $\sigma^r$, respectively\;
\lIf{$l_j=n$}{$\mathrm{leftS:=true}$}\;
\Else
{
    Guess
	$\sigma^{j,0}:=\left(c,(\sigma^{j}_{i})_{i=1}^{l_j-1},(\sigma^{j}_{i})_{i=l_j+1}^{n}  \right)$\;
	$s:=\m{black}(\sigma^{j,0},y,x)$\;
	\lIf{$s=0$}{$\mathrm{leftS:=true}$}\;
	\lElse{$\mathrm{leftS:=false}$}\;
}
\lIf{$\mathrm{leftS}$}{let $a:=1$ and $b:=l_j$}\;
\lElse{let $a:=l_r$ and $b:=n$}\;
$m:=n$ \tcp*{position to be found}
\While{$b>a$}
{	
$l:=\lceil \frac{a+b}{2} \rceil$ \tcp*{position for peg $c$}
\lIf{$\mathrm{leftS}$}{$\sigma^{j,l}:=\left((\sigma^{j}_{i})_{i=1}^{l-1},c,(\sigma^{r}_{i})_{i=l+1}^{l_j},(\sigma^{j}_{i})_{i=l_j+1}^{n}\right)$}\;
\lElse{$\sigma^{j,l}:=\left((\sigma^{r}_{i})_{i=1}^{l_r-1},(\sigma^{j}_{i})_{i=l_r}^{l-1},c,(\sigma^{r}_{i})_{i=l+1}^{n}\right)$}\;
	Guess $\sigma^{j,l}$\;
	$s:=\m{black}(\sigma^{j,l},y,x)$\;
	\If{$s>0$}
	{
	    $b:=l-1$\;
	    \lIf{$b<m$}{let $m:=b$}\;
	}
	\lElse{$a:=l$}\;
}
Return $m$\;
\caption{Function {\sc findNext}}\label{findNext}
\end{algorithm}
The binary search for the exact value of $m$ is done in the interval $[a,b]$, where $m$ is initialized as $n$ and $[a,b]$ as
\[
[a,b]:=\begin{cases}
[1,l_j] & \text{if}\;\,\m{leftS}\\
[l_r,n] & \text{else}
\end{cases}
\]
(lines 9 to 11 of {\sc findNext}).
In order to determine if there is an open correct component on the left side of the current center $l$ of $[a,b]$ in $\sigma^j$ we can define a case dependent permutation:
\[
\sigma^{j,l}:=\begin{cases}
\left((\sigma^{j}_{i})_{i=1}^{l-1},c,(\sigma^{j}_{i})_{i=l}^{l_j-1},(\sigma^{j}_{i})_{i=l_j+1}^{n}\right) & \hspace{-3mm}\text{if}\;\,\m{leftS}\\
\left((\sigma^{r}_{i})_{i=1}^{l_r-1},(\sigma^{r}_{i})_{i=l_r+1}^{l},c,(\sigma^{r}_{i})_{i=l+1}^{n}\right) & \,\text{ else}
\end{cases}
\]
In the first case, the first $l-1$ components of $\sigma^{j,l}$ coincide with those of $\sigma^j$.
The remaining components of $\sigma^{j,l}$ cannot coincide with the corresponding components of the secret code if they have not been fixed, yet.
This is because the $l$-th component of $\sigma^{j,l}$ has the already fixed value $c$, components $l+1$ to $l_j$ coincide with the corresponding 
components of $\sigma^r$ which satisfies $\m{black}(\sigma^r,y,x)=0$ and the remaining components have been checked to be wrong in this case.
Thus, there is a correct open component on the left side of $l$ in $\sigma^j$, if and only if $\m{black}(\sigma^{j,l},y,x)\ne 0$.
In the second case, the same holds for similar arguments.
Now, if there is a correct open component to the left of $l$, we update the binary search interval $[a,b]$ by $[a,l-1]$ and set $m:=\min(m,l-1)$.
Otherwise, we update $[a,b]$ by $[l,b]$.\par

{\bf The Main Algorithm.} The main algorithm is outlined as Algorithm \ref{findAll}.
\begin{algorithm}
Let $y$ be the secret code and set $x:=(0,0,\dots,0)$\;
Guess the permutations $\sigma^i$, $i\in[n-1]$\;
Initialize $v\in\{0,1,\dots,n\}^n$ by $v_i:=\mathrm{black}(\sigma^i,y)$, $i\in[n-1]$, $v_n:=n-\sum_{i=1}^{n-1}v_i$\;
\If{$v=\mathds{1}_n$}
{
    $j:=1$\;
	Find the position $m$ of the correct peg in $\sigma^1$ by at most $\frac{n}{2}+1$ further guesses\;
}
\Else
{
    Call {\sc findFirst} for an active $j\in[n]$ to find the position of the correct peg in $\sigma^j$ by at most $2\lceil\log_2{n}\rceil$ further guesses\;
}
$x_m:=\sigma^j_m$\;
$v_j:=v_j-1$\;
\While{$|\{i\in[n]\,|\,x_i=0\}|>2$}
{
	Choose an active index $j\in[n]$\;
	$m:=\mbox{\sc findNext}(y,x,j)$\;
	$x_m:=\sigma^{j}_m$\;
	$v_{j}:=v_{j}-1$\;
}
Make at most two more guesses to find the remaining two unidentified colors\;
\caption{Algorithm for Permutations}\label{findAll}
\end{algorithm}
It starts with an empty partial solution and finds the components of the secret code $y$ one-by-one.
Herein, the vector $v$ does keep record about the number of open components in which the permutations $\sigma^1,\dots,\sigma^n$ 
equal $y$ and is, thus, initialized by $v_i:=\m{black}(\sigma^i,y)$, $i\in[n-1]$ and $v_n:=n-\sum_{i=1}^{n-1}v_i$.
As mentioned above, the main loop always requires an active index.
For that reason, if $v=\mathds{1}_n$ in the beginning, we fix one solution peg in $\sigma^1$ and update $x$ and $v$, correspondingly.
Every call of {\sc findNext} in the main loop augments $x$ by a correct solution value.
Since one call of $\m{findNext}$ requires at most $1+\lceil\log_2{n}\rceil$ guesses, Algorithm \ref{findAll} does not need more than $(n-3)\lceil\log_2{n}\rceil+\frac{5}{2}n-1$ 
queries (inclusive at most $\frac{n}{2}+1$ initial and $2$ final queries, respectively) to break the secret code.\par

{\bf Example.} We consider the case $n=k=8$ and suppose that the secret code y is
\[
	7\qquad 1\qquad 4\qquad 3\qquad 2\qquad 8\qquad 5\qquad 6
\]
Figure \ref{fig:cyc} shows $n$ possible initial queries.
\begin{figure}[h]
\begin{center}
\begin{tikzpicture}
\draw(4.5,2.5)node[rectangle,draw=black,fill=black!30,minimum height=4.5mm, minimum width=80mm]{};
\draw(4.5,2.0)node[rectangle,draw=black,fill=black!10,minimum height=4.5mm, minimum width=80mm]{};
\foreach \y in {1,2,...,8}
{
	\draw(-0.2,4-0.5*\y) node{$\sigma^{\y}$};
	\foreach \x in {1,2,...,8}
	{
		\pgfmathtruncatemacro{\z}{mod(int(\x+\y-2),8)+1}%
		\draw(\x,4-0.5*\y) node{\z};
	}
}
\draw(10.2,3.5) node{$0\quad\; 0\quad\; 0$};
\draw(10.2,3.0) node{$2\quad\; 0\quad\; 2$};
\draw(10.2,2.5) node{$3\quad\; 2\quad\; 1$};
\draw(10.2,2.0) node{$1\quad\; 1\quad\; 0$};
\draw(10.2,1.5) node{$0\quad\; 0\quad\; 0$};
\draw(10.2,1.0) node{$0\quad\; 0\quad\; 0$};
\draw(10.2,0.5) node{$1\quad\; 0\quad\; 1$};
\draw(10.2,0.0) node{$1\quad\; 0\quad\; 1$};
\draw(4.5,4.25) node{queries};
\draw(10.2,4.25) node{\bf $n_1\;\; n_2\;\; n_3$};
\end{tikzpicture}
\end{center}
\caption{Initial queries $\sigma^j$ with associated responses
$n_1=\m{black}(\sigma^j,y)$,
coincidences with a partial solution $n_2=\m{black}(\sigma^j,x)$,
and the difference of both $n_3$.\label{fig:cyc}}
\end{figure}
We illustrate the procedure {\sc findNext} and further suppose that we
have already identified the positions of 3 colors indicated in the partial solution $x$:
\[
\bullet\qquad \bullet\qquad \bullet\qquad \bullet\qquad 2\qquad \bullet\qquad 5\qquad 6
\]
From the $n_3$ values in Figure \ref{fig:cyc}
we see that $\m{black}(\sigma_3,y,x)=1$ and $\m{black}(\sigma_4,y,x)=0$,
so we choose $3$ as our active index applying {\sc findNext}
with the highlighted initial queries, $\sigma^3$ and $\sigma^4$.
Choosing the already identified color $2$ as a pivot color, {\sc findNext}
does its binary search to identify the next correct peg as demonstrated
in Figure \ref{fig:search}.
\begin{figure}
\begin{center}
\begin{tikzpicture}
\draw(-0.2,1.0) node{$\sigma^a$};
\draw(-0.2,0.5) node{$\sigma^b$};
\draw(-0.2,0.0) node{$\sigma^c$};
\draw(3.0,1.0)node[rectangle,draw=black,fill=black!10,minimum height=4.5mm, minimum width=29mm]{};
\draw(1.5,0.5)node[rectangle,draw=black,fill=black!30,minimum height=4.5mm, minimum width=19mm]{};
\draw(4.0,0.5)node[rectangle,draw=black,fill=black!10,minimum height=4.5mm, minimum width=9mm]{};
\draw(1.0,0.0)node[rectangle,draw=black,fill=black!30,minimum height=4.5mm, minimum width=9mm]{};
\draw(3.5,0.0)node[rectangle,draw=black,fill=black!10,minimum height=4.5mm, minimum width=19mm]{};
\draw(1.0,1.0)node{2}; \draw(2.0,1.0)node{7}; \draw(3.0,1.0)node{8}; \draw(4.0,1.0)node{1};
\draw(1.0,0.5)node{7}; \draw(2.0,0.5)node{8}; \draw(3.0,0.5)node{2}; \draw(4.0,0.5)node{1};
\draw(1.0,0.0)node{7}; \draw(2.0,0.0)node{2}; \draw(3.0,0.0)node{8}; \draw(4.0,0.0)node{1};
\foreach \y in {1,2,3}
{
	\draw(6.5,1.5-0.5*\y)node[rectangle,draw=black,fill=black!30,minimum height=4.5mm, minimum width=40mm]{};
	\foreach \x in {3,4,5,6}
	{
		\pgfmathtruncatemacro{\z}{mod(int(\x+\y),8)+1}%
		\draw(\x+2,1.5-0.5*\y) node{\x};
	}

}
\draw(10.2,1.0) node{$2\quad\; 2\quad\; 0$};
\draw(10.2,0.5) node{$3\quad\; 2\quad\; 1$};
\draw(10.2,0.0) node{$3\quad\; 2\quad\; 1$};
\draw(4.5,1.75) node{queries};
\draw(10.2,1.75) node{\bf $n_1\;\; n_2\;\; n_3$};
\end{tikzpicture}
\end{center}
\caption{Binary search queries to extend the partial solution.
The highlighted subsequences correspond to the subsequences
of the selected initial queries.\label{fig:search}}
\end{figure}
Since the information $n_3$ for query $\sigma^a$ is $0$
(cf. lines 5-7 of Algorithm \ref{findNext})
all correctly placed pegs in $\sigma^3$ are on the left side of the pivot peg.
Thus, we can apply a binary search for the left most correct peg
in the first $4$ places of query $\sigma^3$ using the pivot peg.
here, the binary search is done by queries $\sigma^b$ and $\sigma^c$
and identifies the peg with color $7$
(in general, the peg that is left to the most left pivot position
for which $n_3$ is non-zero).
If the response to $\sigma^a$ would have been greater than $0$,
we would have found analogously a new correct peg in $\sigma^3$
on the right side of the pivot peg.


\subsection{The Case $k>n$}
Now, we consider the variant of Black-Peg Mastermind where $k>n$
and color repetition is forbidden.
Let $y=(y_1,\dots,y_n)$ be the code that must be found.
We use the same notations as above.

{\bf Phase 1.} Consider the $k$ permutations $\overline{\sigma}^1,\dots,\overline{\sigma}^k$, where $\overline{\sigma}^1$ corresponds to the 
identity map on $[k]$ and for $j\in[k-1]$, we obtain $\overline{\sigma}^{j+1}$ from $\overline{\sigma}^{j}$ by a circular shift to the right.
We define $k$ codes $\sigma^1,\dots,\sigma^k$ by $\sigma^j=(\overline{\sigma}^j_i)_{i=1}^n$, $j\in[k]$.
Within those $k$ codes, every color appears exactly once at every position and, thus, we have
\[
\sum_{j=1}^{k}\m{black}(\sigma^j,y)=n,
\]
similar to (\ref{infosum}).
Since $k>n$, this implies that
\begin{lemma}\label{existsWrongQuery}\label{exists0}
There is a $j\in[k]$ with $\m{black}(\sigma^j,y)=0$.
\end{lemma}

{\bf Phase 2.} Having more colors then holes, we can perform our binary search for a next
correct position without using a pivot peg.
The corresponding simplified version of {\sc findNext} is outlined as Algorithm \ref{findNext2}.
\begin{algorithm}[h]
\SetKwInOut{Input}{input}\SetKwInOut{Output}{output}
\Input{Code $y$, partial solution $x\ne 0$ and an active index $j\in[k]$}
\Output{Position $m$ of a correct open component in $\sigma^j$}
\lIf{$j=n$}{$r:=1$ }\;\lElse{$r:=j+1$}\;
$a:=1$, $b:=n$\;
$m:=1$ \tcp*{position to be found}
\While{$b>a$}
{	
	$l:=\lceil \frac{a+b}{2} \rceil$ \tcp*{mid position of current interval}
	Guess $\sigma:=\left((\sigma^{r}_{i})_{i=1}^{l-1},(\sigma^{j}_{i})_{i=l}^{n}\right)$\;
	$s:=\m{black}(\sigma,y,x)$\;
	\If{$s>0$}
	{
	    $a:=l$\;
	    \lIf{$a>m$}{let $m:=a$}\;
	}
	\lElse{$b:=l-1$}\;
}
Return $m$\;
\caption{Function {\sc findNext} for $k>n$}\label{findNext2}
\end{algorithm}
Using that version of {\sc findNext} also allows to simplify our main algorithm 
(Algorithm \ref{findAll}) by adapting lines 2 and 3,
and, due to Lemma \ref{exists0}, skipping lines 4-10.
Thus, for  the required number of queries to break the secret code we have:
the initial $k-1$ guesses,
a call of the modified {\sc findNext} for every but the last two positions
(at most $\lceil\log_2{n}\rceil$ guesses per position)
and one or two final guesses.
This yields, that the modified Mastermind Algorithm
breaks the secret code in at most $(n-2)\lceil\log_2{n}\rceil+k+1$ queries.


\section{Lower Bounds on the Number of Queries}
In the following we consider the case that the secret code has no repetition but arbitrary questions are allowed.
Note that the lower bounds for that case especially hold true
for AB-Mastermind and Permutation-Mastermind, respectively,
since the codebreaker will not be able to detect a secret code with less attempts,
if the set of allowed queries is restricted to the corresponding subset.
Similar to the upper bounds, we proof the respective lower bounds
on the necessary number of queries by construction.

\subsection{The Case $k=n$: Permutation-Mastermind}
Notice that the achieved bound for the case $k=n$ especially holds for Permutation-Mastermind.
In each iteration, the worst case for the code breaker is simulated
by allowing the code maker to replace his secret code with
another permutation from the remaining feasible search space.
For $m\in\mathds{N}$ we denote the $m$-th query of the code breaker with $x^m$
and the $m$-th secret code adaption of the code maker with $y^m$.
The remaining feasible search space $R_m$ consists of all
permutations that agree with the first $m$ pairs of queries and answers:
\[
R_m := \{\sigma\in S_n \;|\; \forall j\in[m]: \black(y^j,x^j)=\black(\sigma,x^j)\}.
\]
Now, a simple strategy of the code maker is to reply every query $x^m$,
$m\in\mathds{N}$, with the smallest possible number
\[
b_m := \min_{\sigma \in R_{m-1}} \black(\sigma,x^m),
\]
choosing his new secret code $y^m\in R_{m-1}$ such that $\black(y^m,x^m)=b_m$.
We obtain our lower bound on the necessary number of queries
by proving the following
\begin{lemma}
It holds that $b_m\le m$ for all $m\in\mathds{N}$.
\end{lemma}
In particular, non of the first $n-1$ queries will be answered with $n$.
Thus, the secret code can not be identified with less than $n$ queries.
\begin{proof}
Assuming that our claim is wrong,
we fix the smallest number $m\in[n]$ with $b_m>m$.
Let
\[
D:=\{c\in[n] \;|\; (x^m)^{-1}(c)=(y^m)^{-1}(c)\}
\]
be the set of colors that are
correctly placed in the current query with respect to the current secret code.
For every $i\in [n]$ let $C_{i}\subseteq [n]$
be the set of all colors that do not occur at position $i$
in any of the former $m-1$ queries nor in the current secret code, i.e.,
\[
C_{i}:=\{c\in[n] \;|\; c \ne x^{\ell}(i)\text{ for all }\ell\in[m]\}.
\]
The intersections $C_{i}\cap D$, $i\in[n]$, are not empty
since $|D|=b_m\ge m+1$
but at most $m$ of the $n$ colors are missing in $C_{i}$.
This fact will enable us to determine a new feasible secret code $z\in R_{m-1}$
such that $\black(z,x^j)=b_j$ for all $j\in[m-1]$
but $\black(z,x^m)<b_m$, a contradiction to the minimality of $b_m$.
The new secret code $z$ is constructed from $y^m$
by changing the colors of some components that coincide with $x^m$,
choosing the new color at a given position $i$ from $C_i \cap D$.
The precise procedure is outlined as Algorithm \ref{algorithm:lowerBound}.
\begin{algorithm}[h]
Set $s:=1$ and $A:=\emptyset$\;
Choose position $i_1\in[n]$ with $y^m(i_1)=x^m(i_1)$\;
Choose color $c_1\in C_{i_1}\cap D$\;
\While{$c_s\not\in A$}
{
	$A:=A\cup\{y^m(i_s)\}$\;
	$s:=s+1$\;
	Define position $i_s:=(y^m)^{-1}(c_{s-1})$\;
	Choose color $c_s\in C_{i_s}\cap D$\;
}
Find the unique $t<s$ with $y^m(i_t)=c_s$\;
$z:=y^m$\;
\lFor{$\ell:=t$ to $s$}{$z(i_\ell):=c_\ell$}
\caption{Secret code adaption, $k=n$}\label{algorithm:lowerBound}
\end{algorithm}
Starting with any position $i_1$ where $y^m$ and $x^m$ have the same color,
we choose another color $c_1\in C_{i_1}\cap D$.
Since $c_1\in D$,
there must be another position $i_2$ such that $y^m(i_2)=c_1=x^m(i_2)$.
Thus, for $s>1$ we can iteratively determine positions $i_s$
where $y^m$ and $x^m$ have the same color, $c_{s-1}$,
and choose a new color $c_s\in C_{i_s}\cap D$ (While loop, lines 4--8).
The iteration stops, if the chosen color $c_s$ corresponds
with a color that appears in $y^m$ at some position $i_t$,
$t<s$, that has been considered before (indicated by the set $A$).
Note, that the iteration must terminate with $2\le s\le m+1$,
since $A$ is empty in the beginning, and $|D|=m+1$.
The set of chosen colors $\{c_\ell \;|\; t \le \ell \le s\}$ is equal
to the set of colors $\{y^m(i_\ell) \;|\; t \le \ell \le s\}$
at the corresponding positions in $y^m$.
Hence,
the new secret code $z$ (defined in lines 10--11) is again a permutation.
Now, let $j\in[m-1]$ be the number of some former query.
Due to the minimal choice of $m$ we have $\black(z,x^j)\ge b_j$.
But $\black(z,x^j)\le b_j$ does also hold
since $\black(y^m,x^j)=b_j$ ($y^m\in R_m$)
and for each position $i$ with $z(i)\ne y^m(i)$
we have $z(i)\ne x^j(i)$ ($z(i)\in C_i$).
Further, the construction of $z$ immediately yields that $\black(z,x^m)<b_m$.
Thus, $z$ is indeed a secret permutation in $R_{m-1}$ that contradicts
the minimality of $b_m$.
\end{proof}

\subsection{The Case $k>n$}
Considering the case $k>n$ we adapt the code maker strategy from the former subsection, i.e. in each turn $m$ the code maker chooses the new secret
code $y^m$ such that the answer is the smallest possible answer $b_m$. We easily obtain a lower bound of $k$ queries by the following 
\begin{lemma}
 It holds that $b_m<n$ for all $m<k.$
\end{lemma}
\begin{proof}
Assume for a moment that there exists an $m<k$ with $b_m=n.$ Like before, let
\[
C_{i}:=\{c\in[n] \;|\; c \ne x^{\ell}(i)\text{ for all }\ell\in[m]\}.
\]
Similar to Algorithm \ref{algorithm:lowerBound} we now replace certain entries of $y^m$ by elements of the corresponding $C_i$.
The detailed procedure is described in Algorithm \ref{algorithm:lowerBound:k>n}.
\begin{algorithm}[h]
Set $s:=1$ and $i_1:=1$\;
Set $A:=\emptyset$ and $B:=\{c\in[k]|\forall i\in[n]: y^k(i)\ne c\}$\;
Choose color $c_1\in C_1$\;
\While{$c_s\not\in A\cup B$}
{
	$A:=A\cup\{y^m(i_s)\}$\;
	$s:=s+1$\;
	Define position $i_s:=(y^m)^{-1}(c_{s-1})$\;
	Choose color $c_s\in C_{i_s}$\;
}
\If{$c_s\in A$}
{
    Find the unique $t<s$ with $y^m(i_t)=c_s$\;
}
\Else
{
   Set $t:=1$\;
}
$z:=y^m$\;
\lFor{$\ell:=t$ to $s$}{$z(i_\ell):=c_\ell$}
\caption{Secret code adaption, $k>n$}\label{algorithm:lowerBound:k>n}
\end{algorithm}

We start with position one and choose a color $c_1\in C_{i_1}$.
As soon as we have $c_s\in B$, we construct $z$ by starting with $y^k$ and then
replacing the color $y^m(i_\ell)$ by the color $c_{i_\ell}$ for any $\ell\le s.$ 
The set of chosen colors $\{c_\ell \;|\; \ell \le s\}$ is equal
to the set of colors $\{y^m(i_\ell) \;|\; \ell \le s\}$ except for  $c_s$ which only appears in the first set
and $y^m(i_\ell)$ which only appears in the second. Since $c_s\in B$ we know that $z$ has no color occurring twice.\par
If the Iteration stops because of $c_s\in A$ the procedure is identic to the one in Algorithm \ref{algorithm:lowerBound}.
So in both cases we find that $\black(z,x^m)<b_m$ and $\black(z,x^\ell)=b_\ell$ for any $\ell\in[m-1]$, in contradiction 
to the minimality of $b_m$.
\end{proof}

\section{Conclusions and Further Work}
In this paper we presented a deterministic algorithm for the identification of a secret code in \enquote{Permutation Mastermind} and \enquote{Black-Peg AB-Mastermind} with more colors than positions.
A challenge of these Mastermind variants is that no color repetition is allowed for a query while most strategies for other Mastermind variants exploit the property of color repetition.
Furthermore we improved the recent lower bound of Berger et al. \cite{BCS16} and showed that the worst case number of queries for Permutation Mastermind is at least $n$, another matter than the asymptotic bound of $O(n)$, which is long-established.
Ko and Teng \cite{KT86} conjecture that this number is actually $\Omega(n\log{n})$, a proof of which would close the gap to the upper bound.
The lower bound proof of Berger et al. is derived by solely considering the search space partition with respect to the number of coincidences with the very first query.
On the other hand, our algorithmic proof does not exploit any structure property of the remaining search space.
For both reasons we expect at least some room for improvements of the lower bound. 
In the future we will take both bounds in focus but the real challenge is to prove or disprove the conjecture of Ko and Teng.



\end{document}